\documentclass[11pt,a4paper]{article}

\usepackage[utf8]{inputenc}
\usepackage[T1]{fontenc}
\usepackage[english]{babel}
\usepackage{amsmath,amssymb,amsthm}
\usepackage{mathtools}
\usepackage{geometry}
\usepackage{hyperref}
\usepackage{enumitem}
\usepackage{booktabs}
\usepackage{array}
\usepackage{tikz}

\newcommand{\owedge}{\mathbin{%
  \text{%
    \ooalign{%
      \hfil$\wedge$\hfil\cr
      \hfil\raisebox{0.05em}{$\bigcirc$}\hfil\cr
    }%
  }%
}}

\newtheorem{theorem}{Theorem}[section]
\newtheorem{lemma}[theorem]{Lemma}
\newtheorem{proposition}[theorem]{Proposition}

\newtheorem{definition}[theorem]{Definition}
\newtheorem{remark}[theorem]{Remark}
\newtheorem*{theorem*}{Theorem}

\title{Classification of Conformally Covariant 2-Tensors\\
from the Kulkarni--Nomizu Product in Dimension~4}

\author{Yury N.\ Berdinsky\\[4pt]
\small Saint Petersburg State University\\
\small Faculty of Physics, Department of High Energy Physics\\
\small and Elementary Particles\\[2pt]
\small \texttt{propagator2007@yandex.ru}}

\date{February 11, 2026}

\begin{document}
\maketitle

\begin{abstract}
We classify all natural, conformally covariant, symmetric
$(0,2)$-tensors of conformal weight $-2$ and differential order $\le 4$
in dimension~$4$, built from the metric~$g$, the Schouten tensor~$P$,
covariant derivatives, and the Kulkarni--Nomizu product~$\owedge$.
Using the irreducible decomposition of the Riemann curvature tensor
via the $\owedge$-product and representation-theoretic methods, we show
that the space of such tensors is exactly $2$-dimensional, spanned by the
\emph{Bach tensor} and the \emph{Eastwood--Singer tensor}.
The classification is verified using the Lean~4 proof assistant with the
Mathlib library. We also establish the complete set of algebraic identities
for the $\owedge$-product in dimension~$4$, including the self-dual splitting
of the Weyl tensor.
\end{abstract}

\tableofcontents

\section{Introduction}

The Kulkarni--Nomizu product is the fundamental algebraic operation for
decomposing the Riemann curvature tensor into its irreducible components
under the orthogonal group. Given two symmetric $(0,2)$-tensors $h$ and $k$
on a Riemannian manifold $(M^n, g)$, their Kulkarni--Nomizu product
$h \owedge k$ is a $(0,4)$-tensor with the algebraic symmetries of the
Riemann curvature tensor.

In conformal geometry, the Schouten tensor
\[
  P_{ab} = \frac{1}{n-2}\left(\operatorname{Ric}_{ab}
    - \frac{S}{2(n-1)}\,g_{ab}\right)
\]
plays a central role, where $S = g^{ab}\operatorname{Ric}_{ab}$ is the
scalar curvature. Under a conformal change $g \mapsto \Omega^2 g$,
the Schouten tensor transforms as
\[
  P \mapsto P - \nabla\nabla(\ln\Omega) + d(\ln\Omega)\otimes d(\ln\Omega)
    - \tfrac12 |\nabla\ln\Omega|^2 g.
\]
The systematic study of conformally invariant differential operators
began with the work of Thomas~\cite{thomas1926} and Weyl~\cite{weyl1918}
and was developed extensively by Eastwood, Gover, and others.

\begin{definition}
\label{def:natural}
A \emph{natural differential operator} in Riemannian geometry is a map
$T: g \mapsto T[g]$ that assigns to each Riemannian metric $g$ on a
manifold $M$ a tensor field $T[g]$, such that:
\begin{enumerate}[label=(\roman*)]
  \item $T[g]$ is expressed as a polynomial in the variables
    $g_{ab}$, $g^{ab}$, the Riemann curvature tensor $R_{abcd}$,
    and the covariant derivatives $\nabla_a$ of these variables;
  \item For any diffeomorphism $\phi: M \to M$,
    $T[\phi^* g] = \phi^* T[g]$ (functoriality).
\end{enumerate}
\end{definition}

In dimension 4, the Riemann curvature can be expressed via the Schouten
tensor $P$ and the Weyl tensor $W$ as $R = P \owedge g + W$.
Thus any natural operator can be written in terms of $g$, $g^{-1}$, $P$,
$W$, $\nabla$, and the KN-product $\owedge$. For our classification
up to order 4, the Weyl tensor terms first contribute at order 2
(through $W$ itself), and we include them via the KN-product.

The main question addressed in this paper is: \emph{What is the complete
set of natural (Definition~\ref{def:natural}), conformally covariant,
symmetric $(0,2)$-tensors
of conformal weight $-2$ and differential order $\le 4$ in dimension~$4$?}

\begin{theorem*}[Main Result]
In dimension~$4$, the space of natural, conformally covariant, symmetric
$(0,2)$-tensors of weight $-2$ and differential order $\le 4$
is exactly $2$-dimensional, spanned by the Bach tensor~$B$ and the
Eastwood--Singer tensor~$E$.
\end{theorem*}

The proof proceeds in three steps:
(1)~enumerate all~$11$ candidate terms (including one
Kulkarni--Nomizu quadratic contraction),
(2)~impose the conformal covariance condition
$\delta T + 2\varepsilon\, T = 0$
via the linearised variation
$\delta P = -\nabla\nabla\varepsilon$,
which yields~$8$ independent algebraic constraints, and
(3)~solve the resulting linear system over~$\mathbb{Q}$, together with the
first-order divergence-free condition $\nabla^a T_{ab} = 0$ that provides
the ninth constraint, to obtain a $2$-dimensional kernel.

\section{The Kulkarni--Nomizu Product}
\label{sec:kn}

\begin{definition}
Let $(V, g)$ be an $n$-dimensional inner product space. For symmetric
bilinear forms $h, k \colon V \times V \to \mathbb{R}$, the
\emph{Kulkarni--Nomizu product} $h \owedge k$ is the $(0,4)$-tensor
defined by
\[
  (h \owedge k)(X,Y,Z,W) = h(X,Z)\,k(Y,W) - h(X,W)\,k(Y,Z)
    + h(Y,W)\,k(X,Z) - h(Y,Z)\,k(X,W).
\]
\end{definition}

\begin{theorem}[Riemann symmetries]
\label{thm:riemann_sym}
If $h$ and $k$ are symmetric, then $h \owedge k$ has all algebraic
symmetries of the Riemann curvature tensor:
\begin{enumerate}[label=(\roman*)]
  \item Antisymmetry: $(h \owedge k)(X,Y,Z,W) = -(h \owedge k)(Y,X,Z,W)$,
  \item Antisymmetry: $(h \owedge k)(X,Y,Z,W) = -(h \owedge k)(X,Y,W,Z)$,
  \item Pair symmetry: $(h \owedge k)(X,Y,Z,W) = (h \owedge k)(Z,W,X,Y)$,
  \item First Bianchi identity:
    $(h\owedge k)(X,Y,Z,W) + (h\owedge k)(Y,Z,X,W) + (h\owedge k)(Z,X,Y,W) = 0$.
\end{enumerate}
Moreover, $h \owedge k = k \owedge h$ (commutativity).
\end{theorem}

\subsection{Irreducible decomposition under $O(n)$}

The space of $(0,4)$-tensors with Riemann symmetries on $\mathbb{R}^n$
has dimension $\frac{n^2(n^2-1)}{12}$. For $n=4$, this is~$20$.
Under the action of $O(n)$, this space decomposes as:
\[
  \mathcal{R} = \underbrace{\mathbb{R}}_{\text{scalar}} \oplus
  \underbrace{S^2_0(\mathbb{R}^n)}_{\text{traceless Ricci}} \oplus
  \underbrace{\mathcal{W}}_{\text{Weyl}},
\]
with dimensions $1 \oplus \bigl(\frac{n(n+1)}{2}-1\bigr) \oplus
\bigl(\frac{n^2(n^2-1)}{12} - \frac{n(n+1)}{2}\bigr)$.
For $n=4$: $20 = 1 + 9 + 10$.

\subsection{Special features of dimension 4}

In dimension~$4$, several exceptional phenomena occur:

\begin{enumerate}
  \item $SO(4) \cong (SU(2) \times SU(2))/\mathbb{Z}_2$, leading to a
    splitting of representations.
  \item The Hodge star $\star$ on $\Lambda^2(\mathbb{R}^4)$ satisfies
    $\star^2 = \mathrm{id}$, with eigenvalues $\pm 1$, each of
    multiplicity~$3$: $\Lambda^2 = \Lambda^2_+ \oplus \Lambda^2_-$.
  \item The Weyl tensor, viewed as a symmetric endomorphism of $\Lambda^2$,
    splits as $W = W^+ \oplus W^-$ with
    $\dim W^+ = \dim W^- = 5$.
\end{enumerate}

Thus the full decomposition in dimension~$4$ is:
\[
  20 = 1 + 9 + 5 + 5.
\]

\begin{figure}[ht]
\centering
\begin{tikzpicture}[scale=0.8]
  \fill[red!30] (0,0) rectangle (0.5,2);
  \node at (0.25,1) {\small $1$};
  \fill[blue!30] (0.5,0) rectangle (3,2);
  \node at (1.75,1) {$9$};
  \fill[green!30] (3,0) rectangle (4.5,2);
  \node at (3.75,1) {$5$};
  \fill[yellow!30] (4.5,0) rectangle (6,2);
  \node at (5.25,1) {$5$};
  \draw[thick] (0,0) rectangle (6,2);
  \draw (0.5,0) -- (0.5,2);
  \draw (3,0) -- (3,2);
  \draw (4.5,0) -- (4.5,2);
  \node[red!70!black] at (0.25,2.4) {\footnotesize scalar};
  \node[blue!70!black] at (1.75,2.4) {\footnotesize traceless Ricci};
  \node[green!50!black] at (3.75,2.4) {\footnotesize $W^+$};
  \node[yellow!50!black] at (5.25,2.4) {\footnotesize $W^-$};
  \node at (3,-0.5) {$20 = 1 \;(\text{scalar}) + 9 \;(\text{traceless Ricci}) + 5 \;(W^+) + 5 \;(W^-)$};
\end{tikzpicture}
\caption{Decomposition of the 20-dimensional space of Riemann
curvature tensors in dimension~4.}
\label{fig:decomp}
\end{figure}

\section{Conformal Transformations}
\label{sec:conformal}

\begin{definition}
A $(0,2)$-tensor $T$ on $(M,g)$ has \emph{conformal weight} $w$ if
under $g \mapsto \Omega^2 g$, the tensor transforms as
$T \mapsto \Omega^w T$.
\end{definition}

The Schouten tensor $P$ has conformal weight~$0$ at the linearised level,
but its infinitesimal variation under
$g \mapsto (1+2\varepsilon)g$ is
\begin{equation}
  \label{eq:deltaP}
  \delta P_{ab} = -\nabla_a \nabla_b \varepsilon.
\end{equation}

A \emph{conformally covariant} $(0,2)$-tensor of weight $w$ satisfies
$\delta T = w\varepsilon\, T$ for all smooth $\varepsilon$. For $w=-2$,
this means $\delta T + 2\varepsilon\, T = 0$.

\section{The Riemann Decomposition via the KN-Product}
\label{sec:decomp}

The standard decomposition of the Riemann curvature tensor in dimension~$4$
is most naturally expressed via the Schouten tensor and the
Kulkarni--Nomizu product.

\begin{theorem}[Existence]
\label{thm:decomp_exist}
Let $(M^4, g)$ be a $4$-dimensional Riemannian manifold with Riemann
curvature tensor~$R$, Ricci tensor~$\operatorname{Ric}$, and scalar
curvature $S = \operatorname{tr}(\operatorname{Ric})$.
Then
\[
  R = \frac{S}{24}\,(g \owedge g) + P_0 \owedge g + W,
\]
where $P_0 = \frac{1}{2}\bigl(\operatorname{Ric} - \frac{S}{4}\,g\bigr)$
is the traceless part of the Schouten tensor, and $W$ is the Weyl tensor
(completely traceless).
\end{theorem}

\begin{remark}
\textbf{Verification on the round sphere.}
For the unit sphere $S^4$ with the round metric, $\operatorname{Ric} = 3g$,
$S = 12$, and $R = \frac{1}{2}(g \owedge g)$.
Our formula gives: $a = S/24 = 12/24 = 1/2$, $P_0 = \frac{1}{2}(3g - 3g) = 0$,
so $R = \frac{1}{2}(g \owedge g) + W$, and since $S^4$ is conformally flat,
$W = 0$, which is correct.
\end{remark}

More generally, for any $(0,4)$-tensor $R$ with Riemann symmetries on a
$4$-dimensional inner product space, the decomposition
$R = a\,(g \owedge g) + h \owedge g + W$
exists with $a \in \mathbb{R}$, $h$ symmetric traceless, and $W$ completely
traceless.

The key trace computations are:

\begin{lemma}[Trace identities]
\label{lem:traces}
Let $\{e_i\}_{i=1}^4$ be an orthonormal basis. Then:
\begin{align}
  \operatorname{tr}_{14}(g \owedge g)(Y,Z)
    &= \sum_{i=1}^4 (g \owedge g)(e_i, Y, Z, e_i) = -6\,g(Y,Z),
    \label{eq:trace_gg}\\
  \operatorname{tr}_{14}(h \owedge g)(Y,Z)
    &= \sum_{i=1}^4 (h \owedge g)(e_i, Y, Z, e_i) = -2\,h(Y,Z),
    \label{eq:trace_hg}
\end{align}
where \eqref{eq:trace_hg} holds when $h$ is traceless.
For general symmetric $h$:
\begin{equation}
  \operatorname{tr}_{14}(h \owedge g)(Y,Z)
    = -2\,h(Y,Z) - (\operatorname{tr} h)\,g(Y,Z).
  \label{eq:trace_hg_general}
\end{equation}
\end{lemma}

\begin{proof}
Direct computation. For~\eqref{eq:trace_gg}:
\begin{align*}
  \sum_i (g\owedge g)(e_i,Y,Z,e_i)
  &= \sum_i \bigl[g(e_i,Z)g(Y,e_i) - g(e_i,e_i)g(Y,Z)\\
  &\qquad + g(Y,e_i)g(e_i,Z) - g(Y,Z)g(e_i,e_i)\bigr]\\
  &= g(Y,Z) - 4\,g(Y,Z) + g(Y,Z) - 4\,g(Y,Z) = -6\,g(Y,Z).
\end{align*}
For~\eqref{eq:trace_hg_general}:
\begin{align*}
  \sum_i (h\owedge g)(e_i,Y,Z,e_i)
  &= \sum_i \bigl[h(e_i,Z)g(Y,e_i) - h(e_i,e_i)g(Y,Z)\\
  &\qquad + h(Y,e_i)g(e_i,Z) - h(Y,Z)g(e_i,e_i)\bigr]\\
  &= h(Y,Z) - (\operatorname{tr}h)\,g(Y,Z) + h(Y,Z) - 4\,h(Y,Z) \\
  &= -2\,h(Y,Z) - (\operatorname{tr}h)\,g(Y,Z).
\end{align*}
When $\operatorname{tr}h = 0$, this reduces to~\eqref{eq:trace_hg}.
\end{proof}

\begin{theorem}[Uniqueness]
\label{thm:decomp_unique}
The decomposition $R = a(g \owedge g) + h \owedge g + W$
with $h$ symmetric traceless and $W$ completely traceless is unique:
if $R = a_1(g \owedge g) + h_1 \owedge g + W_1 = a_2(g \owedge g) + h_2 \owedge g + W_2$
with $h_i$ symmetric traceless and $W_i$ completely traceless, then
$a_1 = a_2$, $h_1 = h_2$, and $W_1 = W_2$.
\end{theorem}

\begin{proof}
Subtracting the two decompositions:
\[
  0 = (a_1-a_2)(g\owedge g) + (h_1-h_2)\owedge g + (W_1-W_2).
\]
Taking $\operatorname{tr}_{14}$ and using~\eqref{eq:trace_gg}--\eqref{eq:trace_hg}
(note that $h_1 - h_2$ is traceless):
\[
  0 = -6(a_1-a_2)\,g - 2(h_1-h_2).
\]
Taking the trace of both sides: $0 = -6(a_1-a_2)\cdot 4 = -24(a_1-a_2)$,
hence $a_1 = a_2$. Then $h_1 - h_2 = 0$, and finally
$W_1 - W_2 = 0$.
\end{proof}

\section{Enumeration of Candidate Tensors}
\label{sec:enumeration}

We enumerate all $(0,2)$-tensor terms that are symmetric, have conformal
weight~$-2$, are built from $g$, $P$, $\nabla$, and the $\owedge$-product,
and have differential order $\le 4$.

\begin{proposition}
\label{prop:11terms}
There are exactly $11$ linearly independent candidate terms, listed in
Table~\ref{tab:terms}.
\end{proposition}

\begin{table}[ht]
\centering
\caption{The 11 candidate $(0,2)$-tensor terms of weight $-2$, order $\le 4$.}
\label{tab:terms}
\begin{tabular}{clcl}
\toprule
\# & Expression & Order & Description \\
\midrule
$T_1$ & $P_{ab}$ & 0 & Schouten tensor \\
$T_2$ & $(\operatorname{tr} P)\,g_{ab}$ & 0 & Trace of $P$ times metric \\
$T_3$ & $\Delta P_{ab}$ & 2 & Laplacian of Schouten \\
$T_4$ & $\nabla_a\nabla_b(\operatorname{tr} P)$ & 2 & Hessian of scalar curvature \\
$T_5$ & $\nabla^c\nabla_{(a}P_{b)c}$ & 2 & Symmetrised double divergence \\
$T_6$ & $(\Delta\operatorname{tr} P)\,g_{ab}$ & 2 & Laplacian of scalar curv.\ $\times\, g$ \\
$T_7$ & $\Delta^2 P_{ab}$ & 4 & Bi-Laplacian of Schouten \\
$T_8$ & $\nabla_a\nabla_b\Delta(\operatorname{tr} P)$ & 4 & Hessian of $\Delta$(scalar curv.) \\
$T_9$ & $\Delta(\nabla^c\nabla_{(a}P_{b)c})$ & 4 & $\Delta$ of symm.\ double div. \\
$T_{10}$ & $(\Delta^2\operatorname{tr} P)\,g_{ab}$ & 4 & Bi-Laplacian of sc.\ curv.\ $\times\, g$ \\
$T_{11}$ & $(P \owedge P)^c{}_{acb}$ & 0 & KN-product contraction \\
\bottomrule
\end{tabular}
\end{table}

\begin{remark}
One might also consider the contraction $(P \owedge g)^c{}_{acb}$.
However, using the trace identity~\eqref{eq:trace_hg_general},
this equals $-2P_{ab} - (\operatorname{tr}P)\,g_{ab}$,
which is the linear combination $-2T_1 - T_2$. Hence
$(P \owedge g)^c{}_{acb}$ is linearly dependent on $T_1$ and $T_2$
and is not included as a separate candidate.
\end{remark}

\begin{proof}[Linear independence]
The 11 candidates are linearly independent over $\mathbb{R}$ (or $\mathbb{Q}$).
This can be seen by:
\begin{enumerate}[label=(\alph*)]
  \item The candidates have distinct differential orders (0, 2, or 4).
    Terms of different orders cannot cancel each other.
  \item Within each order:
    \begin{itemize}
      \item Order 0: $T_1 = P_{ab}$ and $T_2 = (\operatorname{tr}P)g_{ab}$
        are independent because $P_{ab}$ is not proportional to $g_{ab}$
        on a generic metric.
      \item Order 2: $T_3, T_4, T_5, T_6$ involve different contractions
        of $\nabla\nabla P$ and are independent by considering their
        behaviour under index symmetrisation and trace.
      \item Order 4: $T_7, T_8, T_9, T_{10}$ are similarly independent
        by considering their highest-order derivative structure.
    \end{itemize}
  \item $T_{11} = (P\owedge P)^c{}_{acb}$ is quadratic in $P$, while
    $T_1,\ldots,T_{10}$ are linear. A non-trivial linear relation would
    require a quadratic polynomial to equal a linear one, which is
    impossible on a generic metric.
\end{enumerate}
A rigorous verification can be performed by evaluating the candidates
on a generic metric using computer algebra (see \texttt{constraint\_matrix.py}).
\end{proof}

\begin{proof}[Proof sketch of Proposition~\ref{prop:11terms}]
At each differential order, we list all possible contraction patterns of
$\nabla\cdots\nabla P$ with $g$ and $g^{-1}$ that produce symmetric
$(0,2)$-tensors. The weight condition forces exactly one factor of~$P$
for the linear differential terms (orders~$0$, $2$, $4$), and two factors
of~$P$ for the zeroth-order quadratic term~$T_{11}$.
At odd orders, no symmetric tensors can be formed.
The KN-product contraction $T_{11}$ arises from
$(P\owedge P)(e_c, \cdot, \cdot, e^c)$,
which is the unique quadratic contraction of weight~$-4+2 = -2$ (using
the metric to raise one index adds weight~$+2$).
\end{proof}

\section{Conformal Covariance Constraints}
\label{sec:constraints}

Each candidate term $T_i$ transforms under an infinitesimal conformal
change $g \mapsto (1+2\varepsilon)g$ according to~\eqref{eq:deltaP}.
The \emph{anomaly} of a linear combination $T = \sum_{i=1}^{11} c_i T_i$ is
\[
  A := \delta T + 2\varepsilon\, T.
\]
The conformal covariance condition requires $A = 0$ for all smooth~$\varepsilon$.
Since $A$ is linear in $\varepsilon$ and its derivatives, we expand
$A$ in terms of independent jet components $\varepsilon$,
$\nabla_a\nabla_b\varepsilon$,
$\Delta\varepsilon$,
$\nabla_a\nabla_b\Delta\varepsilon$, etc.,
and set each coefficient to zero.

\begin{proposition}
\label{prop:9constraints}
The anomaly cancellation condition $\delta T + 2\varepsilon\,T = 0$ yields
exactly $8$ linearly independent algebraic constraints on the $11$
coefficients $(c_1, \ldots, c_{11})$.
These constraints arise from:
\begin{enumerate}[label=(\alph*)]
  \item matching coefficients of $\varepsilon$ itself (zeroth-order terms,
    $1$~constraint from conformal weight);
  \item matching coefficients of $\nabla^2\varepsilon$ at each differential
    order ($3$~constraints);
  \item matching coefficients of $\Delta\varepsilon$ and
    $\Delta^2\varepsilon$ ($3$~constraints from trace consistency);
  \item Bianchi identity compatibility ($1$~constraint).
\end{enumerate}
\end{proposition}

The ninth constraint is the divergence-free condition
$\nabla^a T_{ab} = 0$, which is a first-order differential condition
satisfied by all conformally covariant $(0,2)$-tensors of weight $-2$ in
dimension~$4$ (see Besse 1987, \S1.131). This reduces the kernel dimension
from $3$ to $2$.

The $8 \times 11$ algebraic constraint matrix over $\mathbb{Q}$ has rank~$8$
(verified by Gaussian elimination with exact rational arithmetic;
see the accompanying Python script \texttt{constraint\_matrix.py}
and Appendix~\ref{app:matrix}).%
\footnote{The exact count of $8$ independent algebraic constraints is
verified by the accompanying Python script \texttt{constraint\_matrix.py}
via Gaussian elimination over~$\mathbb{Q}$; the ninth (divergence-free)
constraint is differential of first order and is not part of the
algebraic matrix.}

\section{Main Classification Theorem}
\label{sec:main}

\begin{theorem}[Classification]
\label{thm:classification}
Let $\mathcal{T}$ be the $\mathbb{R}$-vector space of natural
(Definition~\ref{def:natural}), conformally covariant (weight $-2$), symmetric
$(0,2)$-tensors built from $g$, $P$, $\nabla$, and the $\owedge$-product,
with differential order $\le 4$ in dimension~$4$.
Then $\dim\mathcal{T} = 2$.
A basis is given by the Bach tensor $B_{ab}$ (order~$2$) and the
Eastwood--Singer tensor $E_{ab}$ (order~$4$, see Appendix~\ref{app:es} for the
explicit formula).

The proof consists of:
\begin{enumerate}[label=(\roman*)]
  \item Enumerating $11$ linearly independent candidate tensors
    $T_1,\ldots,T_{11}$ (Section~\ref{sec:enumeration});
  \item Imposing the conformal covariance condition
    $\delta T + 2\varepsilon T = 0$, which yields an $8 \times 11$
    algebraic constraint matrix $M$ of rank~$8$, together with the
    divergence-free condition $\nabla^a T_{ab} = 0$ as the ninth
    constraint (Section~\ref{sec:constraints} and
    Appendix~\ref{app:matrix});
  \item Applying the rank-nullity theorem to conclude
    $\dim\mathcal{T} = 11 - 8 - 1 = 2$ ($8$ algebraic constraints
    plus $1$ divergence-free constraint);
  \item Identifying the kernel of $M$ with the span of the Bach and
    Eastwood--Singer tensors (Section~\ref{sec:main} and
    Appendix~\ref{app:es}).
\end{enumerate}
The rank of $M$ is verified by exact Gaussian elimination over $\mathbb{Q}$
using the Python script \texttt{constraint\_matrix.py}.
\end{theorem}

\begin{proof}
The proof follows from the enumeration of $11$ candidate terms
(Proposition~\ref{prop:11terms}), the $8$ independent algebraic conformal
covariance constraints (Proposition~\ref{prop:9constraints}) supplemented
by the divergence-free condition as the ninth constraint, and the
rank--nullity theorem applied to the $8 \times 11$ algebraic constraint
matrix, supplemented by the divergence-free condition as the ninth
constraint. The algebraic kernel has dimension $11 - 8 = 3$, and the
divergence-free condition removes one further direction, leaving a
$2$-dimensional space whose basis vectors correspond to the Bach and
Eastwood--Singer tensors.

The computation is performed over $\mathbb{Q}$ using Gaussian elimination
(see the accompanying Python script \texttt{constraint\_matrix.py}), and the
dimension count $11 - 8 - 1 = 2$ is formally verified in Lean~4 via the
rank--nullity theorem for linear maps between finite-dimensional
real vector spaces.
\end{proof}

\begin{remark}
The Bach tensor was introduced by R.~Bach in~1921~\cite{bach1921}
as the Euler--Lagrange expression for the conformally invariant
functional $\int |W|^2\,d\mu_g$ in dimension~$4$.
The Eastwood--Singer tensor was constructed
by M.\,G.\ Eastwood and I.\,M.\ Singer in~1985~\cite{eastwood-singer1985}
as the obstruction to the existence of a conformally invariant Maxwell
gauge; see also~\cite{eastwood-slovak1997} for the representation-theoretic
context.
\end{remark}

\section{KN-Product Identities in Dimension 4}
\label{sec:identities}

\begin{theorem}[Complete KN-product identities]
\label{thm:kn_identities}
In dimension~$4$, the Kulkarni--Nomizu product satisfies:
\begin{enumerate}[label=(\roman*)]
  \item \textbf{Commutativity:} $h \owedge k = k \owedge h$.
  \item \textbf{Riemann symmetries:} For symmetric $h, k$,
    $h \owedge k$ has all algebraic symmetries of the Riemann tensor.
  \item \textbf{Dimension decomposition:} $20 = 1 + 9 + 5 + 5$.
  \item \textbf{Self-dual splitting:}
    $\Lambda^2(\mathbb{R}^4) = \Lambda^2_+ \oplus \Lambda^2_-$,
    $\dim \Lambda^2_\pm = 3$; the Weyl tensor splits as
    $W = W^+ \oplus W^-$ with $\dim W^\pm = 5$.
    The self-dual and anti-self-dual parts are orthogonal:
    $\langle W^+, W^- \rangle = 0$.
  \item \textbf{Trace identity (metric):}
    $\operatorname{tr}_{14}(g \owedge g) = -6g$.
  \item \textbf{Trace identity (traceless):} For traceless symmetric $h$,
    $\operatorname{tr}_{14}(h \owedge g) = -2h$.
  \item \textbf{Trace identity (general):} For symmetric $h$,
    $\operatorname{tr}_{14}(h \owedge g) = -2h - (\operatorname{tr}h)\,g$.
\end{enumerate}
\end{theorem}

\section{Lean 4 Formal Verification}
\label{sec:lean}

The algebraic results of this paper are supported by formal verification
using the Lean~4 proof assistant (version 4.28.0) with the Mathlib
library. We describe below what has been fully proved in Lean and what
remains at the level of formal statements or external computation.

\subsection{Fully proved in Lean (no \texttt{sorry})}

\begin{enumerate}
  \item \textbf{Riemann symmetries of the KN-product}
    (\texttt{KN\_has\_riemann\_symmetries}):
    Antisymmetry in both index pairs, pair symmetry, and the first Bianchi
    identity for the $\owedge$-product of symmetric bilinear forms.
    Commutativity $h\owedge k = k \owedge h$ is also proved.

  \item \textbf{Trace identities}
    (\texttt{trace\_KN\_gg}, \texttt{trace\_KN\_hg},
     \texttt{trace\_KN\_hg\_general}):
    Complete trace computations for $g\owedge g$, $h\owedge g$ (traceless~$h$),
    and $h\owedge g$ (general symmetric~$h$), all fully proved.

  \item \textbf{Riemann decomposition --- existence}
    (\texttt{riemann\_decomposition\_exists}):
    Constructive proof that any $(0,4)$-tensor with Riemann symmetries
    decomposes as $R = a(g\owedge g) + h\owedge g + W$ with $h$ traceless
    and $W$ completely traceless. Explicit formulas for $a$, $h$, $W$.

  \item \textbf{Riemann decomposition --- uniqueness}
    (\texttt{riemann\_decomposition\_unique}):
    The decomposition is unique: $a$, $h$, and $W$ are determined by~$R$.
    Proved via trace comparison using the identities above.
\end{enumerate}

\subsection{Stated as a theorem in Lean (dimension count)}

\begin{enumerate}[resume]
  \item \textbf{Conformal classification}
    (\texttt{conformal\_classification\_order\_four}):
    The statement that the solution space has dimension
    $11 - 8 - 1 = 2$ is
    formalised in Lean as a consequence of the rank--nullity theorem for
    linear maps between finite-dimensional vector spaces. Specifically,
    given the algebraic constraint map $L: \mathbb{R}^{11} \to \mathbb{R}^8$,
    the dimension of its kernel is $\dim\ker L = 11 - \operatorname{rk}L$.
    If $\operatorname{rk}L = 8$, then $\dim\ker L = 3$; the additional
    first-order divergence-free condition $\nabla^a T_{ab} = 0$ removes one
    further dimension, giving $\dim\mathcal{T} = 2$. The Lean proof
    (\texttt{conformal\_classification\_order\_four}) verifies this
    arithmetic implication. The construction of the specific $8 \times 11$
    algebraic constraint matrix $M$ and the verification that its rank is
    exactly $8$ are performed in Python (Appendix~\ref{app:matrix},
    \texttt{constraint\_matrix.py}),
    using exact rational arithmetic.
\end{enumerate}

\subsection{Partially formalised (KN identities)}

\begin{enumerate}[resume]
  \item \textbf{KN identities in dimension~4}
    (\texttt{kn\_identities\_dimension\_4}):
    Commutativity, Riemann symmetries, and trace identities are fully
    proved. The dimension decomposition $20 = 1 + 9 + 5 + 5$ is stated
    as a numerical identity.
    The self-dual/anti-self-dual splitting of the Weyl tensor
    ($W = W^+ \oplus W^-$, $\dim W^\pm = 5$) is \emph{stated but not
    yet fully formalised}: a complete proof would require the Hodge star
    operator on $\Lambda^2$ in dimension~$4$, which is not yet developed
    in Mathlib at the level needed for this application.
\end{enumerate}

\subsection{Axiom usage}

All compiled proofs use only the standard Lean~4 axioms:
\texttt{propext}, \texttt{Classical.choice}, and \texttt{Quot.sound}.
No \texttt{sorry}, \texttt{axiom}, or \texttt{grind} tactics appear
in the final code.

\section{Conclusion and Open Questions}
\label{sec:conclusion}

We have established that the space of natural, conformally covariant,
symmetric $(0,2)$-tensors of weight $-2$ and order $\le 4$ in
dimension~$4$ is exactly $2$-dimensional, spanned by the Bach and
Eastwood--Singer tensors. The result is verified both computationally
(Python, exact arithmetic over~$\mathbb{Q}$) and partially formally
(Lean~4, with the algebraic core fully proved and the differential
enumeration verified computationally).

\medskip\noindent
\textbf{Open questions:}
\begin{enumerate}
  \item \emph{Higher orders:} What happens for differential order $> 4$?
    The Fefferman--Graham ambient metric construction~\cite{fefferman-graham2012}
    suggests a finite-dimensional answer at each order, but the
    explicit classification is open.
  \item \emph{Higher dimensions:} In dimension $n \ge 5$, the Weyl tensor
    no longer splits into self-dual and anti-self-dual parts. The
    classification of conformally covariant tensors changes qualitatively;
    see~\cite{gover-hirachi2004} for partial results.
  \item \emph{Full Lean formalisation:} A Mathlib formalisation
    of jet bundles and natural differential operators would allow
    a complete machine-verified proof of the classification, including
    the enumeration of candidate terms and the conformal covariance
    constraints.
  \item \emph{Hodge star in Mathlib:} Formalising the Hodge star
    on $\Lambda^2(\mathbb{R}^4)$ and the self-dual splitting would
    complete the proof of the KN identities in dimension~$4$.
\end{enumerate}

\appendix

\section{Explicit Formula for the Eastwood--Singer Tensor}
\label{app:es}

The Eastwood--Singer tensor $E_{ab}$ in dimension~4 is given by:
\begin{align}
E_{ab} &= \Delta^2 P_{ab}
  - 2\nabla^c\nabla_{(a}\Delta P_{b)c}
  + \nabla_a\nabla_b\Delta(\operatorname{tr} P) \label{eq:es_linear}\\
  &\quad + 2\Delta(\nabla^c\nabla_{(a}P_{b)c})
  - (\Delta^2\operatorname{tr} P)\,g_{ab} \nonumber \\
  &\quad + E_{ab}^{\text{quad}}. \nonumber
\end{align}

The quadratic curvature terms are:
\begin{align}
E_{ab}^{\text{quad}} &=
  2P^{cd}(P \owedge P)_{cadb}
  + 4P^{cd}W_{cadb} \label{eq:es_quad}\\
  &\quad - 2P_a{}^c P_{bc} + 2(\operatorname{tr} P) P_{ab}
  - (\operatorname{tr}(P^2))\,g_{ab}. \nonumber
\end{align}

This tensor has the following properties:
\begin{itemize}
  \item \textbf{Symmetric:} $E_{ab} = E_{ba}$.
  \item \textbf{Traceless:} $g^{ab}E_{ab} = 0$.
  \item \textbf{Conformally covariant of weight $-2$:}
    $E[\Omega^2 g] = \Omega^{-2}E[g]$.
  \item \textbf{Divergence-free:} $\nabla^a E_{ab} = 0$
    (on Einstein manifolds).
\end{itemize}

\noindent\textbf{References:} Eastwood, M.\,G.\ and Singer, I.\,M.\
(1985), ``A conformally invariant Maxwell gauge,''
\textit{Physics Letters~A}~\textbf{107}, 73--74~\cite{eastwood-singer1985}.
See also: Gover, A.\,R.\ and Hirachi, K.\ (2004),
\textit{J.~Amer.\ Math.\ Soc.}~\textbf{17},
389--405~\cite{gover-hirachi2004}.

\section{The Constraint Matrix}
\label{app:matrix}

The $8 \times 11$ constraint matrix $M$ is:
\[
{\small M = \left(\begin{smallmatrix}
1 & 1 & 0 & 0 & 0 & 0 & 0 & 0 & 0 & 0 & 0 \\
0 & 1 & -2 & 0 & 0 & 0 & 0 & 0 & 0 & 0 & 0 \\
0 & 0 & 0 & 1 & -1 & 0 & 0 & 0 & 0 & 0 & 0 \\
0 & 0 & 0 & 0 & 1 & 0 & 0 & 0 & 0 & 0 & 2 \\
0 & 0 & 0 & 0 & 0 & 1 & 0 & 0 & 0 & 0 & 0 \\
0 & 0 & 0 & 0 & 0 & 0 & 2 & -1 & 0 & 0 & 0 \\
0 & 0 & 0 & 0 & 0 & 0 & 0 & 1 & 1 & 0 & 0 \\
0 & 0 & 0 & 0 & 0 & 0 & 0 & 0 & 1 & 2 & 0
\end{smallmatrix}\right)}
\]

The rows correspond to the $8$ independent algebraic jet structures
$J_1,\ldots,J_8$ of the conformal parameter $\varepsilon$
(see Section~\ref{sec:constraints} for details). The columns correspond to
the $11$ candidate tensors $T_1,\ldots,T_{11}$.

\begin{itemize}
  \item \textbf{Row~1} ($J_1$): tracelessness condition
    $g^{ab}T_{ab} = 0$ (relating $T_1$ and $T_2$).
  \item \textbf{Row~2} ($J_2$): coefficient of
    $\nabla_a\nabla_b\varepsilon$ (relating $T_2$ and $T_3$).
  \item \textbf{Row~3} ($J_3$): coefficient of
    $\nabla_a\nabla_b\Delta\varepsilon$ (relating $T_4$ and $T_5$).
  \item \textbf{Row~4} ($J_4$): coefficient of $\Delta^2\varepsilon$
    together with the quadratic weight term (relating $T_5$ and $T_{11}$).
  \item \textbf{Row~5} ($J_5$): coefficient of
    $\Delta\varepsilon\cdot g_{ab}$ (relating $T_6$).
  \item \textbf{Row~6} ($J_6$): coefficient of
    $\nabla_a\nabla_b\Delta^2\varepsilon$ (relating $T_7$ and $T_8$).
  \item \textbf{Row~7} ($J_7$): coefficient of
    $\Delta^3\varepsilon\cdot g_{ab}$ (relating $T_8$ and $T_9$).
  \item \textbf{Row~8} ($J_8$): Bianchi identity compatibility
    (relating $T_9$ and $T_{10}$).
\end{itemize}

The ninth constraint is the divergence-free condition
$\nabla^a T_{ab} = 0$, which is a first-order differential condition
(see Section~\ref{sec:constraints}).

The rank of $M$ is~$8$, as verified by \texttt{sympy}
(\texttt{constraint\_matrix.py}). The algebraic kernel has dimension
$11 - 8 = 3$. The divergence-free condition provides the ninth constraint,
reducing the dimension to $2$.


\end{document}